\documentclass[a4paper,12pt]{scrartcl}

\usepackage[utf8]{inputenc}

\usepackage{amsmath,bm,amsthm,graphicx,skmath}

\usepackage{graphicx,xcolor}

\usepackage{bera}
\usepackage[charter]{mathdesign}

\usepackage{csquotes}

\usepackage{booktabs,nicematrix}

\usepackage{cite}
\usepackage{xspace,microtype}

\usepackage{fnpct}

\newtheorem{theorem}{Theorem}

\newcommand{\Balken}[2]{\textcolor{#1}{\rule{#2}{0.7em}}}

\definecolor{darkgreen}{cmyk}{0.5,0.1,0.5,0.5}
\definecolor{midgreen}{rgb}{0.0,0.5,0.0}
\definecolor{midred}{rgb}{0.8,0.0,0.0}
\definecolor{midcyan}{cmyk}{1.0,0.0,0.0,0.4}
\definecolor{orange}{cmyk}{0.0,0.8,1.0,0.0}

\begin{document}

\title{Why Majority Judgement is not yet the solution for political
  elections, but can help finding it} \author{Friedemann Kemm}

\maketitle

\begin{abstract}
  Like many other voting systems, Majority Judgement suffers from the
  weaknesses of the underlying mathematical model: Elections as
  problem of choice or ranking. We show how the model can be enhanced
  to take into account the complete process starting from the whole
  set of persons having passive electoral rights and even the aspect
  of reelection. By a new view on abstentions from voting and an
  adaption of Majority Judgement with three grades, we are able to
  describe a complete process for an election that can be easily put
  into legislation and sets suitable incentives for politicians who
  want to be reelected.
\end{abstract}

\tableofcontents

\section{Introduction}
\label{sec:einleitung}

Although in the last decades many alternative voting methods were
proposed, plurality voting is still in widespread use. While plurality
voting is the best method for proper binary decisions like guilty/not
guilty, motion/proposal accepted/rejected, etc., it has significant
flaws for decisions with more options.
This is especially true for runoffs,  as can be seen, e.\,g.,
by an example Balinski and Laraki presented
in~\cite{balinski:hal-01137173}: Candidates A and B are rated in a
poll with grades \textcolor{midgreen}{Very Good},
\textcolor{midcyan}{Acceptable}, or \textcolor{midred}{Poor},

\begin{NiceTabular}{ll}
  A & \textcolor{green}{\rule{4em}{0.7em}}
  \textcolor{cyan}{\rule{3.5em}{0.7em}} \textcolor{red}{\rule{2.5em}{0.7em}}
  \\
  B & \textcolor{green}{\rule{3.5em}{0.7em}}
  \textcolor{cyan}{\rule{3em}{0.7em}} \textcolor{red}{\rule{3.5em}{0.7em}}
\end{NiceTabular}

where the bars represent fractions of votes. In this case, we would
expect A to win the vote. But that is not guaranteed since the votes
could be distributed as follows:

\begin{NiceTabular}{ll}
  A & \Balken{green}{.5em} \Balken{green}{3.5em}
  \Balken{cyan}{3.5em} \Balken{red}{2.5em}
  \\ 
  B & \Balken{cyan}{.5em} \Balken{red}{3.5em}
  \Balken{green}{3.5em} \Balken{cyan}{2.5em}
\end{NiceTabular}

But the situation is even worse than Balinski and Laraki paint it. The
distribution might be as follows:

\begin{NiceTabular}{ll}
  A & \Balken{green}{3.5em} \Balken{green}{0.5em} \Balken{cyan}{2.5em}
  \Balken{cyan}{1em} \Balken{red}{2.5em}
  \\
  B & \Balken{green}{3.5em} \Balken{cyan}{0.5em} \Balken{cyan}{2.5em}
  \Balken{red}{1em} \Balken{red}{2.5em}
\end{NiceTabular}

In this case, we know that 15\% of the voters prefer A over B. Of the
other 85\%, we know nothing. It would even be possible that all of
them think that both candidates are equally suited for the office. But
they still have to decide between the candidates when presented with
the ballot. Their decision might depend on anything basically
irrelevant like the quality of their morning coffee, the mood of their
spouse, the flip of a coin, etc. Thus, a runoff comes close to a
lottery\footnote{Actually, in ancient Athenian democracy, selecting
  persons for offices by lot was an accepted
  method~\cite{mclean2015strange}.} in some cases and, hence, cannot
be considered a proper binary decision even if both candidates are
considered acceptable by a majority.

This flaw is inherited by all preferential voting systems. Therefore,
it is not surprising that legislature always hesitated to adopt any of
the alternatives suggested in the last decades. Actually, the
beginning of Social Choice, the science of collective decision making,
as a discipline and theory on its own consists of collecting flaws and
discouraging results, most notably the impossibility results by
Arrow~\cite{arrow1950difficulty,arrow2012social},
Gibbard~\cite{gibbard-1973}, and
Satterthwaite~\cite{satterthwaite-1975}.  Regarding democratic
elections, the most important flaws are the No Show Paradox, which
might discourage voters from voting, the dependence on irrelevant
alternatives, which might discourage parties from proposing
candidates, and manipulability, which might discourage voters from
expressing their opinion honestly. The effect of the latter is also
known as strategic or tactical voting. Originally, the research
focused on ranked voting systems. But even early grading based voting
systems such as score methods, where voters could freely assign points
in a certain range to candidates, turned out to be unsatisfying as
they are heavily susceptible to strategic voting. An improvement was
made by Majority Judgement~\cite{book-mj}, where points are replaced
by a merely ordinal scale of grades. While all grading based methods
avoid the dependency on irrelevant alternatives, Majority
Judgement~(MJ) also encourages voters to express their opinion
honestly.

In this study, we point out that a major issue of the applied
theoretical framework is its restricted angle of view. Most often,
elections are considered as a problem of Social Choice. At the one
hand, this is convenient since for Social Choice a full theory
exists~\cite{taylor-pacelli,taylor2005social,sen2018collective,gaertner2009primer,book-mj,zbMATH06414025,zbMATH06414059},
and there are many real world problems to which this theory was
successfully applied\footnote{The problems
  Arrow~\cite{arrow1950difficulty} originally had in mind where very
  different from electing people for political offices.}. On the other
hand, it falls short of providing a guideline for good election
procedures. Most often, it is only applied to the voting at the end of
the selection process. But as illustrated in
Figure~\ref{fig:participation}, the major part of eliminating people
with passive rights of voting from the set of choices is done before
the complete electorate is asked to vote. Furthermore, the available
choices are human beings who may react to the results of the vote and
the election procedure. Considering the fact that the election has to
be repeated every few years, this means that the incentives set by the
election method have a severe impact on the politics between
elections.

\begin{figure}
  \centering
  \includegraphics[width=.5\linewidth]{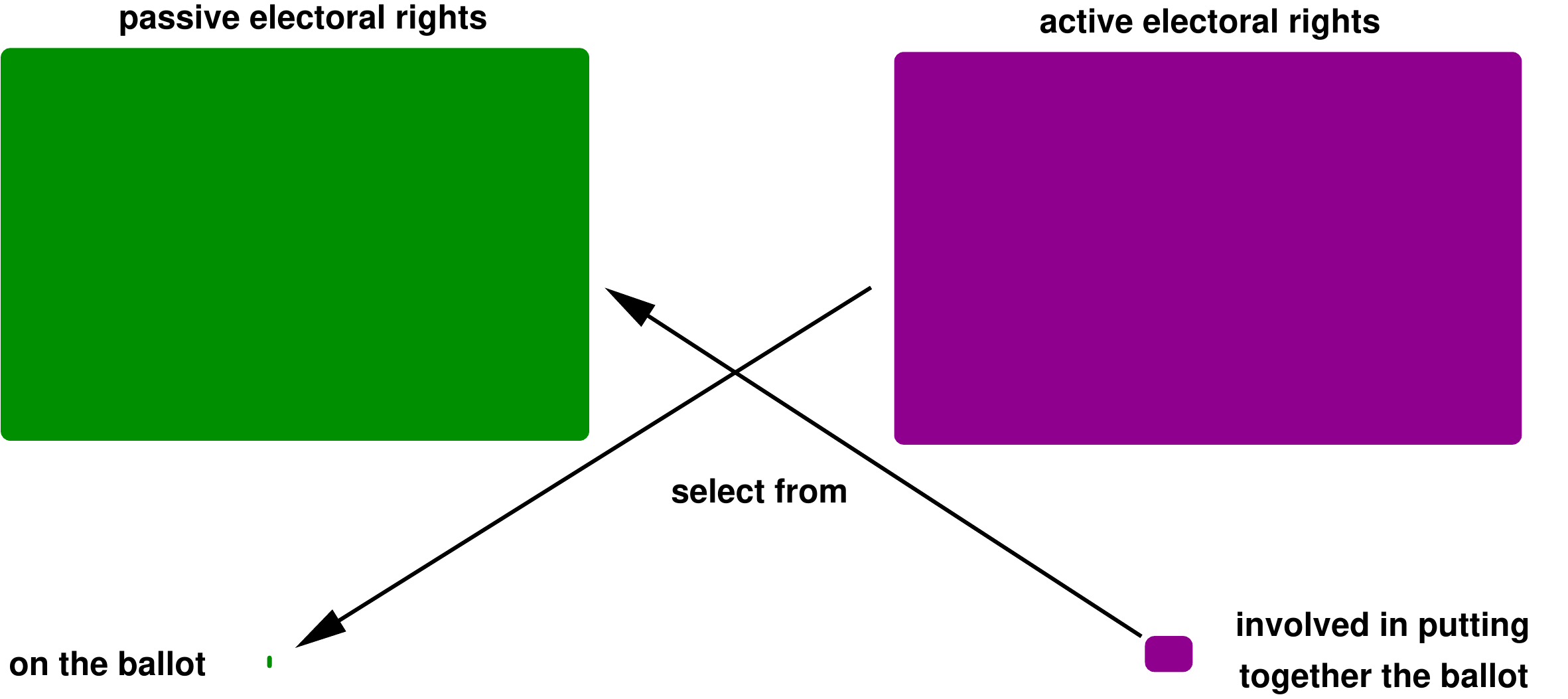}
  \caption{Only a minority makes up the ballot.}
  \label{fig:participation}
\end{figure}

We want to stress that the considerations in this study are solely
focused on democratic elections, electing human beings who can react
to incentives and are typically elected for a short term of only a few
years, but with the possibility of reelection. There are many problems
of collective decision making where our approach would not be the best
process.

Our study is organized as follows: First, we discuss the ways in which
the mathematical model for political elections should be
broadened. For that purpose, we take a closer look at the medieval
contributions by Ramon LLull and Nicolaus Cusanus and how their
methods were fitted to the circumstances. We then present three
criteria an election method for modern democracy should satisfy and
investigate how far a method based on binary decisions would get us.
In Section~\ref{sec:how-major-judg}, we discuss aspects of Majority
Judgement that help us to find a new procedure that satisfies above
criteria. Finally, in Section~\ref{sec:new-view-abstentions}, we
introduce an election procedure that fulfills above criteria. The key
is a new view on abstention from voting combined with the idea that
the final election consists of simultaneously deciding on motions.

\section{Why modelling elections as a problem of choice is
  insufficient}
\label{sec:why-modell-elect}

In this section, we want to investigate which criteria a modern
election procedure should satisfy. Therefore, we describe two medieval
suggestions in order to replace the unsatisfying plurality vote. It
will become clear that back then considering the election as a problem
of choice was still sufficient. But in the meantime, the situation has
changed in two major ways: Since every citizen above a certain age has
passive rights of election, the number of eligible persons is much
higher. Furthermore, available positions are not for a lifetime but
for a certain term, usually with the possibility for reelection. While
the standard theory of social choice does not take into consideration
the latter, the issue of the increased number of eligible persons is
\enquote{solved} in a somehow dirty way: the theory is only applied to
the voting and excludes the process that composes the ballot. Thus,
the problem considered by the theory of social choice is overly
narrowed down and we should look for a broader framework, starting
again from the medieval suggestions.

\subsection{How medieval approaches were superior to our modern
  approach}
\label{sec:how-medi-appr}

Here, we start with Nicolaus Cusanus, since the set of alternatives
from which to choose and the electorate were much smaller than for
Ramon LLull. Both developed methods in the area of preferential
voting\,/\,ranked voting. 

\subsubsection{Nicolaus Cusanus}
\label{sec:nicolaus-cusanus}

In his 1433/34 work \emph{De concordantia catholica}\footnote{The
  whole work can be found, e.\,g., at
  https://urts99.uni-trier.de/cusanus/content/werke.php?werk=9.},
among other things, Cusanus made a suggestion for the election of the
king of the Holy Roman Empire (HRE). Although not explicitly written
down in the constitution (Golden Bull of 1356), it was consensus that
the king had to be one of the prince electors, the
electorate\footnote{There were a few exceptions during history.}. This
board consisted of seven people, which leaves us with a rather small
problem. Thus, it was easy to perform a method different from
plurality voting, which by then was the standard\footnote{\dots and
  since they did not adopt the new method, stayed the same until the
  end of the HRE in 1806.}. Each prince elector would now be given
ballots with the names of the prince electors\footnote{Cusanus
  suggested seven ballots, each with only one name.} and would have to
rank them. The ranks then would be given points and the winner be
decided by summing up the points given by all electors. A modification
of this method became later known as Borda voting or Borda count. The
main difference to the Borda method is the ballot: Nicolaus Cusanus
had all eligible people on the ballot. For modern democracy, this
would be impractical, since in some cases the ballot would have to
feature millions of names. The selection process needed to narrow down
the ballot is usually ignored in order to apply methods provided by
the theory of Social Choice. But this way, the wrong problem is
treated.

\subsubsection{Ramon Llull}
\label{sec:ramon-llull}

Ramon Llull went one step further than Cusanus: In his 1299 treatise
\emph{Ars Electionis}\footnote{In\cite{hagele2001lulls}, a
  reproduction of the relevant passages together with an English
  translation is presented.}, he described in his work a process for
the election of an abbess for a nuns' monastery. They
started with all 20 sisters who were eligible.
By some form of plurality voting, they narrowed down the
group to seven sisters. In order to cope with the weaknesses of
plurality voting, they now had to use plurality voting to add to their
group two more sisters. These nine sisters then had to
perform\footnote{Starting from that point, the description by Llull is
  somehow inconsistent, since the two additional sisters are suddenly
  forgotten.}  what nowadays is known as Copeland's method, a
form of the Condorcet method.

What is most notable about this process is that he considers the
complete selection process starting with the complete set of persons
having passive rights of voting. Furthermore, it shows that Llull was
fully aware of the difficulties large sets of possible choices pose on
more elaborate methods of preferential voting. His solution at this
point was for the first steps to return to plurality voting. For
modern democratic elections, this is still not feasible. We again
would have one step with possibly millions of names on the ballot.

\subsection{How terms and the possibility of reelection introduced a
  new challenge}
\label{sec:how-terms-poss}

For modern elections, we have to deal with even more issues than the
people in middle ages: In modern democracy there are terms, which
means that there is no lifetime appointment. We also have to deal with
evolution in time, since the voters attitude towards the candidates
might change depending on their performance during former terms. If
for example, we employ approval
voting~\cite[Section~1.8]{taylor-pacelli}, the elected person would
improve their chances for reelection by just avoiding angering
anyone. Since every bold move or decision bears the danger of angering
voters, politics would be in some sense paralyzed. This shows that the
election method has a severe impact on politics and the decision
making of the elected people.

It is worth noting that this is relevant even for lifetime positions:
It is not the choice between possible dishes for the next birthday
party or destinations for future holidays. These possible choices
cannot act on themselves. But people hoping to get elected can act
themselves and with their behavior and decision making in other
positions influence the voters' stance on them.

\subsection{What a modern approach should take into account}
\label{sec:what-modern-approach}

From our discussion of the medieval contributions of Cusanus and
Llull, we can conclude the following desirable properties for modern
election processes and the mathematical modelling thereof:
\begin{description}
\item[Completeness] Take into account the complete process, starting with the
  complete set of eligible persons, i.\,e.\ with all people having
  passive electoral rights. 
\item[Feasibility] Divide the process into steps that are feasible from an
  organizational point of view. 
\item[Appropriate Incentives] Take into account the incentives the
  method sets for the candidates, especially those already in office.
\end{description}
Only considering all three aspects together allows for the
construction of an election process that will be fully accepted by
both, legislature and the voters.

\subsection{How a procedure based on proper binary decisions would look like}
\label{sec:how-procedure-based}

Before we dive into modern methods, let's consider what is possible
using classical binary decisions. However, as mentioned above, the
Condorcet method,
(a)~is based on of runoffs and, thus, non-proper binary
decisions, and~(b) is not feasible due to the large number of persons
having passive suffrage in modern democracy.
Nevertheless, a possible procedure may be comprised of following steps:
\begin{enumerate}
\item Let the voters decide if the ballot is acceptable, or if all
  candidates would need to be rejected. 
\item If the ballot is rejected, search for new candidates.
\item Otherwise, let the voters decide if the upper or lower half of
  the ballot offers the better choice.
\item Narrow the ballot down to the preferred half and repeat these
  steps until the preferred part contains only one candidate. 
\end{enumerate}
In this way, an election could be broken down to a series of
decisions, most of which are proper binary decisions. \enquote{Most}
because the later steps are close to or even a runoff themselves,
which, as we have seen in Section~\ref{sec:einleitung}, is not a
proper binary decision. The first step, however, can in any
case be considered a proper binary decision. Another flaw is that
in each step a bias exists towards the larger part of the
partitioned ballot. This might be overcome by putting one name in both
parts, which, on the other hand, would lead to a bias towards this
candidate. But instead of a further discussion on that point, we want
to check to what extent this procedure would satisfy the above
criteria
\begin{description}
\item[Completeness] Since the first step is modelled after the
  consecutive steps by dividing the set of people with passive
  electoral rights into two parts, namely those on the ballot and
  those not on the ballot, it takes into account the complete set of
  possible choices and, thus, leads to a complete method.
\item[Feasibility] If one would like to perform this in only one round
  of voting, at least if the ballot is accepted, the complexity for
  the voter would then increase with the number of names on the
  ballot. That is, if there are seven names on the ballot, each voter
  would have to first indicate if they approve the ballot. Then they
  would have to indicate which part they prefer, the first part with
  four names or the second part with three names.

  Next, the same for the halves of both the upper and the lower part,
  the latter comprised of one with two names and another one with only
  one name. And then for each of the three pairs left, they would
  indicate which person they would prefer. This means that each voter
  has to place exactly seven ticks or crosses (depending on the
  tradition) on the ballot. While the number of ticks/crosses might be
  acceptable, the pattern how these need to be placed might be
  confusing to some voters. Thus, there is still room for improvement.
\item[Appropriate Incentives] The method does not encourage the
  candidates to behave in a certain way. Neither does it prefer
  candidates that have shown decision-making-power nor does it prefer
  the hesitant. In this respect, it is probably the most neutral
  method one could imagine. While sometimes this might be desirable,
  there are situations where this is not the optimal strategy. Here is
  hope for some improvement although that is not an easy task. 
\end{description}
A major reason why this method is not recommended is the danger of
repeatedly getting the whole ballot rejected which could lead to long
periods of paralysis or even an unwanted person still being in office
while they should already have been replaced by someone else. Note
that this issue cannot be resolved by replacing steps~3 and~4
with another method from Social Choice.

\section{How Majority Judgement introduced a new and useful point of
  view on the problem}
\label{sec:how-major-judg}

Most traditional voting methods require the voters to vote on
individual persons. This means that we have to deal with nominally
scaled attributes. Plurality voting would mean selecting the modal
value. Quite recently, Balinski and Laraki introduced a new
method~\cite{book-mj}, where electing is only a secondary aspect:
Majority Judgement (MJ). In this method, the voters choose grades for
all candidates separately.  Since grades are ordinally scaled, we can
adopt the much more significant median of the grades over the modal
value.  One of the first applications of the method was wine tasting,
where the focus is not on a competition but on the grades for the
wines themselves. With a suitable tie breaking rule, this can also be
used to rank and to elect candidates.

\subsection{Three improvements brought in by Majority Judgement}
\label{sec:two-impr-brought}


The first and most important improvement brought in by Majority
Judgement is the naturally fixed set of possible choices: the
available grades. For ranked voting, the number of possible choices
can vary wildly, for a complete method in the above sense even up to
many millions so that it has to be narrowed down in a first step. On
the other hand, deciding on votes is always feasible if a suitable set
of grades is offered to the voters.

A second improvement is that by judging on grades, there is no
dependency on irrelevant alternatives. Grades cannot withdraw their
existence on the ballot.


A third improvement is the higher quality of scaling. Instead of
nominal scaling, we now deal with ordinal scaling. It is obvious that
a better grade than the median would be rejected by a majority, a
worse grade would be objected by at least half of the electorate. This
removes many of the flaws and drawbacks of plurality voting (like
voting on the name of a project, street, \dots).

While at first Majority Judgement was only intended to give grades as
in wine degustation, where the final grades are later printed on the
labels of the readily bottled wine, we here discuss a problem of
ranking and electing. In fact, Balinski and Laraki offer a tie
breaking rule that allows to rank the candidates with equal grade and,
thus, give a method for electing: by choosing the winner of the
ranking.

\subsection{The favorable properties of MJ with three grades}
\label{sec:favor-prop-mj}

While Balinski and Laraki originally favored a greater number of
available grades, since this would increase the amount and quality of
information collected from the voters, other authors pointed out some
weaknesses of that approach. One of them we illustrate in the
following example: A group of 21 middle school students have to decide
on the destination for a school outing. Two of the possible choices
would be the high ropes course and the zoo. Assume that in the group
are ten eager and enthusiastic students, ten who are the exact
opposite, and one student more on the neutral side. Then the outcome
of Majority Judgement with four grades, \textcolor{midgreen}{Cool!},
\textcolor{midcyan}{Nice}, \textcolor{darkgray}{Ok}, and
\textcolor{midred}{Help, no!}, might be

\begin{NiceTabular}{lccc}
  & enthusiastic & neutral & no motivation \\
  \dots & \dots & \dots & \dots\\
  \textbf{H} & \textcolor{green}{\rule{10em}{1.3em}} &
  \textcolor{gray}{\rule{1em}{1.3em}} &
  \textcolor{gray}{\rule{10em}{1.3em}} \\
  \textbf{Z} & \textcolor{cyan}{\rule{10em}{1.3em}} &
  \textcolor{cyan}{\rule{1em}{1.3em}} &
  \textcolor{red}{\rule{10em}{1.3em}} \\
  \dots & \dots & \dots & \dots
\end{NiceTabular}

Although everybody except one person would prefer the high ropes
course, the zoo gets the better majority grade. If at least one of the
eager students had not shown up, or if two of the non enthusiastic
students had not shown up, the high ropes course would instead get the
better majority grade. This means that MJ suffers from the so called
No Show Paradox\footnote{The No Show Paradox, first described by
  Fishburn and Brams~\cite{doi:10.1080/0025570X.1983.11977044}, is a
  weakness of some voting systems in which it is possible that a group
  of voters might be better of when they had abstained from voting. It
  is automatically prevented by using score methods, i.\,e.\ methods
  where candidates are given points by the voters. The No Show Paradox
  is typical for other
  methods~\cite{MOULIN198853,felsenthal2019voting}. For example, in
  the presence of a large electorate, Condorcet voting with more than
  three alternatives would suffer from the paradox.}. Since it is
obvious that with less grades available situations where preferences
and majority grades conflict are fewer, we should next consider
Majority Judgement with three grades, for which the absence of the No
Show Paradox is already proven~\cite{mjvsav}. Note that lowering the
number of grades to two would resemble approval voting which, as we
have discussed in Section~\ref{sec:how-terms-poss} above, is a poor
choice due to the incentives it sets.

Hence, in the remainder of this section, we assume a scale with the
three grades \emph{positive, neutral, negative}. Although, as we will
see later in Section~\ref{sec:new-view-abstentions}, this is not the
best scale for application to democratic elections, it helps
illustrating the properties of Majority Judgement with three grades.
In~\cite{mjvsav}, Balinski and Laraki also show how in this case the
presentation can be done without explicitly dealing with the majority
grades, but instead with a score. While they formulate the score with
values relative to the total number of votes, for the sake of
convenience, we use the raw form of the score:
\begin{equation}
  \label{eq:1}
  S =
  \begin{cases}
    N_{\text{positive}} & \text{if}\quad N_{\text{positive}} >
    N_{\text{negative}}\;,\\
    - N_{\text{negative}} & \text{otherwise}\;. 
  \end{cases}
\end{equation}
Here, we still might have ties. These are broken by means
of~\(N_{\text{negative}}\) for~\(S_A = S_B > 0\)
and~\(N_{\text{positive}}\) for~\(S_A = S_B \leq 0\).
This can be executed by a tie breaking score, which is in some
sense dual to the score itself:
\begin{equation}
  \label{eq:2}
  T =
  \begin{cases}
    - N_{\text{negative}} & \text{if}\quad N_{\text{positive}} >
    N_{\text{negative}}\;,\\
    N_{\text{positive}} & \text{otherwise}\;. 
  \end{cases}
\end{equation}

The only possible ties left would be the cases where A and B have
identical results. One advantage of resorting to the raw score is that
the theoretical proofs are much simpler, which means, for instance,
that the absence of the No Show Paradox is easily seen:

Assume that we have~\(\tilde S_A > \tilde S_B\) before the last voter
takes action and that the last voter grades~A better than~B. Then it
is obvious that~\(S_A \geq \tilde S_A\) and~\(S_B \leq \tilde S_B\),
which implies~\(S_A > S_B\). The situation
with~\(\tilde S_A = \tilde S_B\) and~A the winner of the tie break, is
treated similarly.

Using the raw form, we can prove a weak form of consistency. We recall
that consistency means that if the electorate is divided in two parts
which both elect the same winner, then this candidate is also the
winner of the complete vote. This is obviously satisfied for all
traditional score methods, i.\,e. methods where voters assign in some
way points to candidates, and the points are then summed up for each
candidate over all votes. The most prominent example would be
plurality voting, where each voter can assign exactly one point to
exactly one candidate. Smith\cite{smithpref} proved that among all
ranked voting systems, the score methods are also the only ones
satisfying the consistency criterion\footnote{To be more precise,
  Smith used a generalized definition of score methods, where also
  secondary, tertiary, etc.\ votes (for tie-breaking) are taken into
  consideration.}. Condorcet, Instant Runoff, Coombs, etc.\ all
violate consistency. This might also be one of the reasons,
legislature never really took up these methods.

As equation~\eqref{eq:1} shows, MJ with three grades is close to the
score methods, the main difference being that we have positive points
(for positive votes) and negative points (for negative votes) that are
not naively summed up. Dependent on the difference between positive
and negative votes only the positive or only the negative points are
summed up. But we still can prove the following weaker form of
consistency:
\begin{theorem}\label{thm-cons}
  Assume the electorate is divided in two parts, each of which would
  elect candidate~A. If in both parts, the score of~\(A\) has the
  same sign, i.\,e.~\(S_A^{(1)} \cdot S_A^{(2)} > 0\)
  or~\(S_A^{(1)} = S_A^{(2)} = 0\), then~A is the winner of the
  vote.
\end{theorem}
\begin{proof}
  Since~\(S_A^{(1)} = S_A^{(2)} = 0\) is only possible if all voters
  grade candidate~A neutral, this case is trivial.
  
  If~\(S_A^{(1)}, S_A^{(2)} > 0\) and~\(S_A^{(1)} > S_B^{(1)},\,
  S_A^{(2)} > S_B^{(2)}\) for any other candidate~B, then
  \begin{equation*}
    S_A = S_A^{(1)} + S_A^{(2)} > \max{S_B^{(1)},0} +
    \max{S_B^{(2)},0} \geq S_B\;.
  \end{equation*}
  The case~\(S_A^{(1)}, S_A^{(2)} < 0\) is even simpler since then
  all candidates have negative scores, and we find
  \begin{equation*}
    S_A = S_A^{(1)} + S_A^{(2)} > S_B^{(1)} +
    S_B^{(2)} = S_B\;.
  \end{equation*}
  In both cases,~A is the winner of the vote.  The remaining case,
  the existence of another candidate~B
  with~\(S_A^{(1)} = S_B^{(1)},\, S_A^{(2)} = S_B^{(2)}\), can be
  treated in a similar manner by using the tie breaking score~\(T\)
  instead of~\(S\). 
\end{proof}
Thus, for all reasonable partitions of the electorate consistency is
satisfied. Note that for preferential voting systems, there is no
possibility to distinguish between reasonable and unreasonable
partitions. An immediate consequence is that situations like the
decision between two destinations for a school outing as described
above, where nearly everybody grades the high ropes course higher than
the zoo, but the zoo wins over the high ropes course, rarely
occur. If, for example, we lump together \emph{nice} and \emph{OK},
and we let the student who had decided the vote against the high ropes
course grade more extreme---the high ropes course negative and the zoo
positive---then the overall vote would look like

\begin{NiceTabular}{lccc}
  & enthusiastic & strange guy & no motivation \\
  \dots & \dots & \dots & \dots\\
  \textbf{H} & \textcolor{green}{\rule{10em}{1.3em}} &
  \textcolor{red}{\rule{1em}{1.3em}} &
  \textcolor{gray}{\rule{10em}{1.3em}} \\
  \textbf{Z} & \textcolor{gray}{\rule{10em}{1.3em}} &
  \textcolor{green}{\rule{1em}{1.3em}} &
  \textcolor{red}{\rule{10em}{1.3em}} \\
  \dots & \dots & \dots & \dots
\end{NiceTabular}

The votes for the single alternatives could then be reordered as

\begin{NiceTabular}{lc}
  \dots & \dots \dots \dots\\
  \textbf{H} & \textcolor{green}{\rule{10em}{1.3em}} 
  \textcolor{gray}{\rule{10em}{1.3em}} 
  \textcolor{red}{\rule{1em}{1.3em}} \\
  \textbf{Z} & \textcolor{green}{\rule{1em}{1.3em}}
  \textcolor{gray}{\rule{10em}{1.3em}} 
  \textcolor{red}{\rule{10em}{1.3em}} \\
  \dots & \dots  \dots \dots
\end{NiceTabular}

giving both alternatives the new middle grade as majority grade and by
the tie breaking rule ranking the high ropes course above the
zoo. Using above defined score, we find~\(S_H = 10 > 0\) and
~\(S_Z = -10\). If we partition the group, our electorate, the only
possibilities for a part that gives H a negative score would be the
strange guy plus at most one enthusiastic student and some of the
others. But then the remaining part would give H a positive
score. Thus, the only partitions with
~\(S_H^{(1)} \cdot S_H^{(2)} > 0\) automatically
yield~\(S_H^{(1)}, S_H^{(2)} > 0\) and~\(S_Z^{(1)}, S_Z^{(2)} < 0\)
ranking the high ropes course above the zoo in both parts. The
conditions of Theorem~\ref{thm-cons} are now exactly fulfilled, and as
we have seen, overall the high ropes course is ranked above the
zoo. It is also easy to see that even if we expand the group of
strange students up to eight, we still can partition the group into
two parts, both giving H a positive score and ranking it above Z.


\subsection{Why other tie breaking rules for Majority Judgement are
  not suited to our purpose}
\label{sec:why-other-tie}

The above defined scores,~\(S\) as defined in equation~\eqref{eq:1}
and the tie breaking score~\(T\) as defined in equation~\eqref{eq:2},
result from the original tie breaking procedure for Majority Judgement
by Balinski and Laraki. Of course, other tie breaking rules for
alternatives with the same majority grade could be applied, e.\,g.\
those presented in~\cite{Fabre2021}. For the above used three grade
scale \emph{positive, neutral, negative} these rules would lead to the
(relative) scores
\begin{align}
  \label{eq:3}
  d & =  \frac{N_\text{positive}-N_\text{negative}}{N_\text{total}}\;,
      \tag{relative difference} \\[\medskipamount]
  s & =
      \frac{N_\text{positive}-N_\text{negative}}{N_\text{positive}+N_\text{negative}}\;,
      \tag{relative share} \\[\medskipamount]
  n & = \frac{N_\text{positive}-N_\text{negative}}{N_\text{neutral}}\;,
      \tag{normalized difference}
\end{align}
Considering incentives, the first two have some similarities with
approval voting: avoiding negative votes at any cost is a strong
strategy. Thus, they are not suited for our purpose. The normalized
difference is much closer to the original tie breaking rule of
Majority Judgement, but lacks its simplicity. Furthermore, unlike the
original rule, it is not strategy proof.  As we will see, the original
rule represented by the scores discussed in
Section~\ref{sec:favor-prop-mj} sets good incentives for the
candidates.  Thus, there is no need for a replacement. Still, a close
look at~\cite{Fabre2021} will reveal that there are many applications
in other collective decision making problems where one or the other of
these rules is favorable. An example would be sports competitions like
figure skating and MJ with a higher number of available grades. But
for democratic elections, it is best to stick to the original tie
breaking rule (and, of course, three grades).

\section{A new view on abstentions from voting as the key to the
  solution}
\label{sec:new-view-abstentions}

While Majority Judgement did a great step in the right direction, it
still is not yet what we would need. Firstly, it does not take into
account the process before the ballot is presented to the
voters. Secondly, it is not completely clear what the minimum
requirements are for a candidate to be elected. As a remedy, we shift
our point of view. Instead of considering the election process as a
problem of choice, we now consider it as deciding on multiple motions
simultaneously.

\subsection{Simultaneously deciding on motions instead of choosing}
\label{sec:simult-decid-modt}

If an issue arises that calls for legislature to take action, i.\,e.\
to pass a new law, usually several formulations are proposed by
different parties; sometimes even several versions by a single
party. The parliament then has to decide on these motions
sequentially, usually starting with the version with the most
extensive consequences. Once one of these motions is approved, the
remaining motions are obsolete, and the decision is made.

If we now think of the ballot as a list of motions \enquote{Elect A
  for the office}, \enquote{Elect B for the office}, \dots, then we
might decide on them simultaneously. Among those candidates who would
be approved, we then choose a winner by means of social choice. For
this last step, we suggest a version of Majority Judgement since it is
compatible with the first step, namely deciding who is
acceptable. Compared to the more traditional process we described in
Section~\ref{sec:how-procedure-based}, the likelihood of getting the
complete ballot rejected is much smaller. Furthermore, we do not have
to deal with the issues and flaws of this method.

For deciding on motions in a board or parliament, there are two
traditional approaches: (1)~Everybody has to approve or disapprove the
motion. (2)~In addition, one can abstain from voting. While (1)
resembles grading with two grades, (2) resembles grading with three
grades. Since grading with two grades when performed simultaneously
would lead to approval voting, which, due to wrong incentives, we
already classified as undesirable for our purpose, (2) can be used to
perform Majority Judgement with three grades. But as mentioned above,
the grades \emph{positive, neutral, negative} or here \emph{approve,
  abstain, disapprove} are not the optimal scale, the issue being
\emph{abstention from voting}.

\subsection{The new view on abstention from voting}
\label{sec:new-view-abstention}

Abstention from voting can have two different meanings: weak approval
or no judgement at all, e.\,g., due to lack of sufficient
information. This makes it difficult to interpret the outcome of a
poll using the categories \emph{approval, abstention,
  disapproval}. When used in an electing procedure, it could even be
disastrous. If a candidate is unknown to most voters except their
supporters, this candidate would get a positive score and could, thus,
even win the election.

Our suggestion is to rearrange the two possible meanings of
abstentions and explicitly differentiate between them. Abstention due
to lack of information should result in a lower grade than weak
approval. Boiled down to three grades, we arrive at \emph{strong
  approval}, \emph{weak approval}, and \emph{no explicit
  approval}. This means that no judgement and disapproval are lumped
together, which is in nice accordance to Balinski and Laraki who in
most of their publications stress out that no grade for an alternative
should be considered as the lowest grade for that alternative.

If we now think of a candidate unknown to the majority of the
electorate, then this person would get the lowest grade from a
majority which, thus, is their majority grade. Since we do not want
such a candidate to be elected, we have to set as minimum criterion
for being elected majority grade better than \emph{no explicit
  approval}, which means at least \emph{weak approval} by a majority
of the electorate.

\subsection{A practical procedure for electing via simultaneously
  deciding on motions}
\label{sec:pract-proc-elect}

\begin{figure}
  \centering
  \fbox{%
  \includegraphics[width=.7\linewidth]{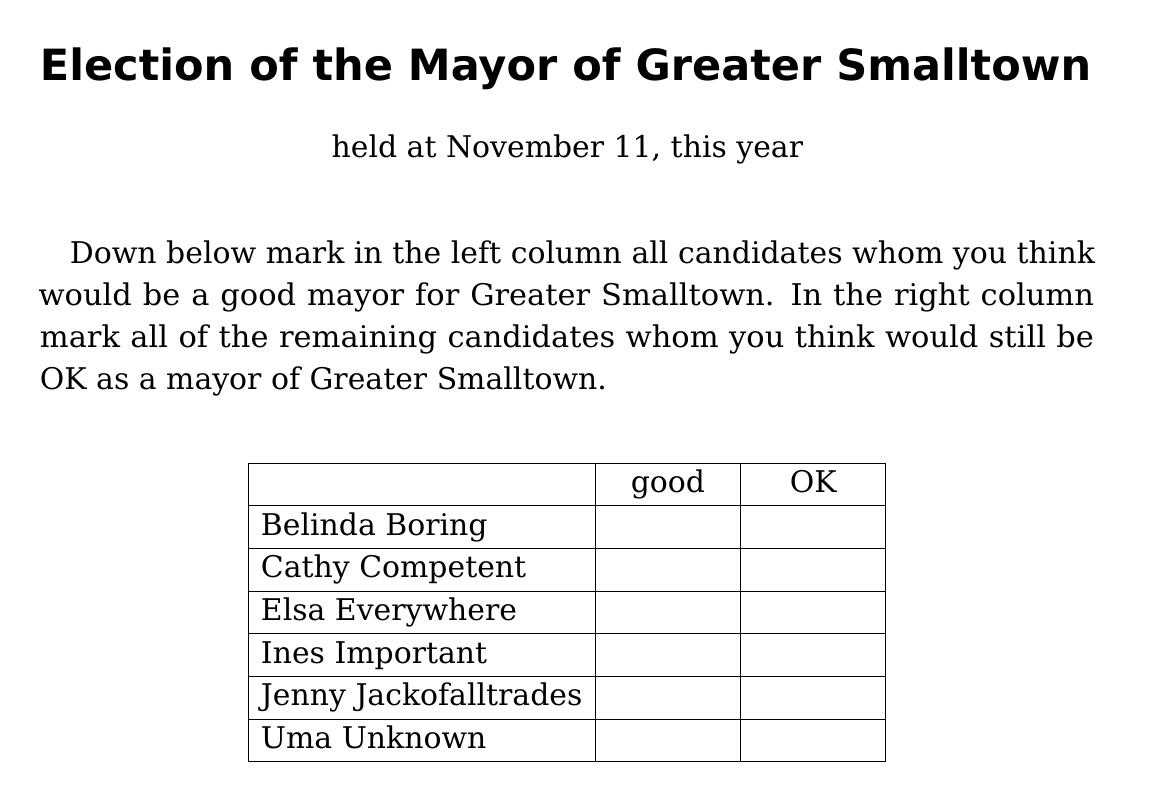}}
  \caption{Stripped down version of a ballot for simultaneously
    deciding on motions. Usually one would provide more information on
    the candidates, e.\,g.\ their party, profession, etc.}
  \label{fig:ballot}
\end{figure}

With this, the first step in political elections, namely the
nomination of candidates needs little to no changes. On the ballot,
voters are requested to first mark in the left column all
candidates whom they strongly approve. Then, they are asked to mark
in the right column all remaining candidates, they still weakly
approve.
The levels of approval should be formulated via grades like \emph{good} and
\emph{OK} or \emph{A--C} and \emph{D--E}. In Figure~\ref{fig:ballot},
we show a rather simple example of a ballot for this procedure.

In the end, for each alternative two scores are determined: The
number of votes with strong approval and the summed up number of votes
with any approval. Here, we deviate from the usual score for Majority
Judgment with three grades: Since we do not ask for negative votes,
there should also be no negative scores.
The results are then determined as follows:
\begin{itemize}
\item If no alternative is approved at least in the weak sense by a
  majority, the election has to be repeated and new candidates need to
  be found.
\item Otherwise, the candidates with a higher number of votes with
  strong approval than votes without explicit approval are ranked
  according to the votes with strong approval.

  Ties are decided by the number of votes with weak approval, i.\,e.\
  by the summed up approval. 
\item Thereafter, the other candidates are ranked according to their
  summed up approval.

  Ties are decided by the number of votes with strong approval. 
\end{itemize}
Since there will be a winner as soon as there is one candidate on the
ballot who is graded as at least weakly approved by a majority of the
electorate, the likelihood of the complete ballot being rejected is
much smaller than for the process described in
Section~\ref{sec:how-procedure-based}. As a consequence, chances for a
party to obstruct the election by nominating only poor candidates,
e.\,g.\ if their favorite is already in office but has no right to be
reelected due to having served the maximum number of terms already,
are close to zero. On the contrary, they are forced to propose good
candidates because otherwise another party's candidate would win the
election. As we will point out in the next section, there are even
more desirable properties of the new procedure.

\subsection{Properties of the new procedure}
\label{sec:prop-new-proc}

Since the new procedure is based on Majority Judgement with three
grades, it shares all of its favorable properties. It does not suffer
from the No Show Paradox, it is weakly consistent
(Theorem~\ref{thm-cons}), strategy proof etc.

Since the method allows for the whole ballot to be rejected, it is
\emph{complete}. As can be seen from above description, it is also
\emph{feasible}. And the ballot will not be very different from a
traditional ballot (see example in Figure~\ref{fig:ballot}). We find
that it is way simpler to fill out than for Instant Runoff, which is
one of the rare modern methods already in use (e.\,g.\ in Australia
and Ireland). Furthermore, the results are represented in a way that
is easily understood by the public, as can be seen in
Figure~\ref{tab:results}. An alternative would be the graphical
representation we used throughout this paper or a combination of
both. It is of course important that every voter can easily comprehend
how the ranking was derived from the votes.

In Figure~\ref{tab:results}, the ranking is divided in several
blocks. The first block (here Rank 1--3) consists of the candidates
with more votes giving them strong approval than no explicit
approval. They are ranked according to the first column. All other
candidates are ranked by the sum of the first and second column or
with increasing number of votes without explicit approval, i.\,e.\
according to the last column. These candidates, again, are divided
into two blocks: those who would be elected in the absence of the
higher ranked candidates and those who would not be elected even in
the absence of all other candidates.

While the given names in the example are mere alliterations to the
surnames, the surnames themselves are chosen as to illustrate the
expected results for the candidates. The example shows that a person
who is unknown to a vast majority of the electorate will never be
elected, just as we desired. As can be seen from the fifth ranked
candidate, it might be fatal to polarize the electorate. On the other
hand, the third and fourth ranked candidates indicate that going for
the opposite is also not a winning strategy.  These considerations
lead us already to the incentives, which the method sets for the
candidates, especially those already in office who want to be
reelected.

\begin{figure}
  \centering
  \fbox{%
    \begin{minipage}{.8\linewidth}
    Ranking of the candidates for the election as a mayor
  of Greater Smalltown. The first ranked candidate is the winner of
  the vote; the results are given in percentages:
  \begin{center}
    \begin{tabular}{|c|l|c|c|c|}
      \hline
      Rank & Candidate & Good & OK & None \\
      \hline\hline
      1 & Cathy Competent &  50 &  20 & 30 \\
      \hline
      2 & Jenny Jackofalltrades & 45 & 35 & 20 \\
      \hline
      3 & Elsa Everywhere & 25 & 60 &  15 \\
      \hline
      \hline
      4 & Belinda Boring & 10 & 80 & 10 \\
      \hline
      5 & Ines Important & 44 &  10 & 46 \\
      \hline\hline
      6 & Uma Unknown & 16 & 1 & 83 \\
      \hline
    \end{tabular}
  \end{center}
  \end{minipage}
}
  \caption{Presentation of election results as a table.}
  \label{tab:results}
\end{figure}

But before we further investigate the \emph{incentives} for the
candidates, we want to discuss the incentives for those who put
together the ballot. It is one of the major purposes of political
parties to find suitable candidates for an office like a mayor, member
of parliament etc. Due to being sensitive to irrelevant alternatives,
plurality voting has shown a tendency to discourage parties from
proposing candidates since that might bear the danger of having
someone from the opposite end of the political spectrum being elected,
especially if other parties from the own side of the spectrum propose
candidates. Thus, the voters are only offered a restricted choice of
candidates. whereas simultaneously deciding on all candidates is more
encouraging to parties to propose candidates. Since tactical
considerations are not applicable to the process, the parties are also
encouraged to propose high quality candidates. Having someone of lower
quality run for office just because they have some merits in the party
is no successful strategy. Also withdrawing the candidacy of someone,
who is already successful in office, just because some people in the
party dislike that person, has much more dangerous consequences than
with plurality voting.

A few of the incentives we have already discussed above at hand of the
example shown in Figure~\ref{tab:results}: Polarizing the electorate
might be fatal, much like the opposite strategy, and without being
known to most of the electorate, there is little chance to win the
vote. We now look at what might happen when a candidate, already in
office and hoping for reelection, is faced with a critical decision
that might polarize the electorate. In order to investigate this
scenario, without the advantage of empirical data, we have to resort
to some kind of symmetry considerations like we usually do in
probability theory, e.\,g.\ for modelling a dice where we assume each
possible result to have the same probability. In the same way, we now
assume that the electorate polarized in a symmetrical way: the
resulting shift from weak approval to strong approval and to no
explicit approval, in this case rejection, is of the same size. This
means that the difference between votes with strong approval and votes
without explicit approval remains unchanged. If the number of votes
with strong approval was larger than of those without explicit
approval, it will stay larger. In this case, which is most likely for
a person already elected in a previous election, the relevant number
is the strongly approving votes and will therefore increase. Thus,
this candidate's chances for reelection increase. People who are
highly popular will be rewarded for bold decisions. This is
illustrated as follows:

\vspace{.5em}

\noindent Before critical decision:\\
\Balken{green}{7em}\Balken{cyan}{8em}\Balken{gray}{5em}

\noindent After critical decision:\\
\Balken{green}{9em}\Balken{cyan}{4em}\Balken{gray}{7em}

\vspace{.5em}

On the other hand, if a person got elected in the previous election
although not being popular, i.\,e.\ although having gotten less votes
with strong than without any approval, they are discouraged from
making polarizing decisions. This is in fact a good property of the
election procedure: Bold decisions will be made in directions that are
approved by the majority rather than in a direction that most people
reject. How strong this desirable effect is has to be studied at hand
of empirical data, e.\,g., once such a process is in practice and
empirical data can be more easily collected. Although it is possible
that the empirical results call for modifications of our procedure, we
are confident that it is already very strong concerning incentives,
and we are sure that in this respect, it is far superior to any other
procedure in use.

\section{Conclusions}
\label{sec:conclusions}

In this paper, we proposed a new method for elections in modern
democratic states. For this purpose, we changed the point of view on
elections in several ways: We took into account the whole set of
theoretically electable people. We took into account the incentives,
the method sets for people in office and wanting to be reelected, and
we renounced the setting of Social Choice by viewing the actual
election at the end of the process as simultaneously deciding on
motions. We also proposed new criteria an election procedure should
satisfy: completeness, feasibility, and good incentives. With
traditional approaches, it is impossible to achieve these at the same
time. Especially, completeness and feasibility would be
contradictory. As we illustrated, the new method is easy to implement,
and the results are easy to understand for the voters. The method at
the same time gives the voters the opportunity to reject the complete
ballot and the political parties good incentives to propose candidates
that will not be rejected. Finally, given incentives set by the
method, people in office like mayors, members of parliaments,
etc. will be encouraged to follow their beliefs instead of tactical
thoughts about the next election.

Since we employ a version of Majority Judgment with three grades for
our procedure, we also looked at the theory of this method and proved
that for reasonable decompositions of the electorate, the method is
consistent. Together with the other desirable properties of MJ with
three grades, namely the absence of the No Show Paradox, independence
of irrelevant alternatives, and robustness against tactical voting,
etc., this makes the suggested process well acceptable for
legislature. Thus, this study might help to motivate legislatures all
over the world to replace plurality voting by modern methods.

\bibliographystyle{amsplain} \bibliography{Wahlen}

\end{document}